\documentclass[conference,letterpaper]{IEEEtran}
\usepackage[utf8]{inputenc}

\newcommand{\myTitle}{Time-Shifted Alternating Gelfand-Pinsker\\ Coding for Broadcast Channels \xspace}

\PassOptionsToPackage{ngerman,british}{babel}
\usepackage{babel}

\usepackage{xspace}
\usepackage{orcidlink}
\usepackage{csquotes}
\usepackage{graphicx}
\usepackage{pgfplots}
\usepackage{tikz}

\usepackage{amsmath}
\usepackage{amssymb}
\usepackage{amsthm}
\usepackage{mathtools}
\usepackage[bb=dsserif]{mathalpha}
\usepackage{bm}
\usepackage{mathrsfs}
\usepackage{mleftright}
\usepackage{siunitx}
\usepackage[only,llbracket,rrbracket]{stmaryrd}
\usepackage{hyperref} 

\usepackage[%
	xindy,
	acronym,
	nopostdot,
	shortcuts,
]{glossaries} %

\pgfplotsset{compat=1.18}
\graphicspath{{images/}{../data/plots/}{data/plots}}
\usepgfplotslibrary{units,fillbetween}
\usetikzlibrary{external,pgfplots.units,arrows.meta,matrix,overlay-beamer-styles,chains,positioning,intersections,calc,fit,dsp,patterns}
\pgfplotsset{table/search path={results/},}
\tikzexternaldisable

\definecolor{wongblue}{RGB}{0, 114, 178}
\definecolor{wongorange}{RGB}{230, 159, 0}
\definecolor{wonggreen}{RGB}{0, 158, 115}
\definecolor{wongpurple}{RGB}{204, 121, 167}
\definecolor{wonglightblue}{RGB}{86, 180, 233}
\definecolor{wongvermillion}{RGB}{213, 94, 0}
\definecolor{wongyellow}{RGB}{240, 228, 66}

\definecolor{matlabblue}{rgb}{     0,     0.447, 0.741}
\definecolor{matlaborange}{rgb}{   0.85,  0.325, 0.098}
\definecolor{matlabyellow}{rgb}{   0.929, 0.694, 0.125}
\definecolor{matlabpurple}{rgb}{   0.494, 0.184, 0.556}
\definecolor{matlabgreen}{rgb}{    0.466, 0.674, 0.188}
\definecolor{matlablightblue}{rgb}{0.301, 0.745, 0.933}
\definecolor{matlabred}{rgb}{      0.635, 0.078, 0.184}

\pgfplotscreateplotcyclelist{wongcolors}{{wongblue}, {wongorange}, {wonggreen}, {wongpurple}, {wonglightblue}, {wongvermillion}, {wongyellow}}

\dspdeclareoperator{dspshapexor}{
	\pgfmoveto{\pgfpointadd{\centerpoint}{\pgfpoint{0pt}{-\radius}}}
	\pgflineto{\pgfpointadd{\centerpoint}{\pgfpoint{0pt}{ \radius}}}

	\pgfmoveto{\pgfpointadd{\centerpoint}{\pgfpoint{-\radius}{0pt}}}
	\pgflineto{\pgfpointadd{\centerpoint}{\pgfpoint{ \radius}{0pt}}}

	\pgfusepathqstroke
}

\tikzset{%
	block/.style     = {draw,rectangle,align=center,inner sep=2mm},
	bigblock/.style  = {draw,rectangle,align=center,inner sep=2mm,minimum height=2.5em},
	hiergroup/.style = {draw,line width=0.3pt,inner sep=5mm,rectangle,rounded corners},
	dspxor/.style    = {shape=dspshapexor,line cap=rect,line join=rect,line width=\dspblocklinewidth,minimum size=\dspoperatordiameter},
}

\pgfplotsset{
	discard if symbolic not/.style 2 args={
		x filter/.code={
			\edef\tempa{\thisrow{#1}}
			\edef\tempb{#2}
			\ifx\tempa\tempb
			\else
				
			\fi
		}
	},
	discard if/.style 2 args={
		x filter/.append code={
				\ifdim\thisrow{#1} pt=#2 pt
						
				\fi
		}
	},
	discard if not/.style 2 args={
			x filter/.append code={
					\ifdim\thisrow{#1} pt=#2 pt
					\else
							
					\fi
			}
	}
}

\makeatletter
\newcommand{\@pltref}[1]{\tikzexternaldisable\ref{#1}\tikzexternaldisable}
\newcommand{\@@pltref}[1]{(\tikzexternaldisable\ref{#1}\tikzexternaldisable)}

\newcommand{\pltref}{\@ifstar\@pltref\@@pltref}
\makeatother

\theoremstyle{plain}
\newtheorem{thm}{Theorem}

\theoremstyle{definition}

\theoremstyle{remark}

\newcommand*{\matr}[1]{\bm{#1}}
\newcommand{\bmat}[1]{\ensuremath{\begin{bmatrix}#1\end{bmatrix}}}

\DeclareMathOperator{\Ber}{Ber}

\DeclareMathOperator{\E}{\mathbb E}

\DeclareMathOperator{\He}{\mathbb H} %
\DeclareMathOperator{\MI}{\mathbb I}

\DeclarePairedDelimiter\abs{\lvert}{\rvert}
\DeclarePairedDelimiter\idxset{\llbracket}{\rrbracket}

\DeclarePairedDelimiter\set{\{}{\}}

\DeclarePairedDelimiterX{\infdivx}[2]{(}{)}{#1\;\delimsize\|\;#2}

\DeclareSIUnit{\belc}{Bc}
\DeclareSIUnit{\belm}{Bm}
\DeclareSIUnit{\bit}{bit}
\DeclareSIUnit{\sample}{S}
\DeclareSIUnit{\bpcu}{bpcu}

\newacronym{ASK}{ASK}{amplitude-shift keying}
\newacronym{AWGN}{AWGN}{additive white Gaussian noise}
\newacronym{DMBC}{DMBC}{discrete memoryless broadcast channel}
\newacronym{BER}{BER}{bit error rate}
\newacronym{BICM}{BICM}{bit-interleaved coded modulation}
\newacronym{biDMC}{biDMC}{binary-input discrete memoryless channel}
\newacronym{bpcu}{bpcu}{bits per channel use}
\newacronym{BPSK}{BPSK}{binary phase-shift keying}
\newacronym{BC}{BC}{broadcast channel}
\newacronym{CCDM}{CCDM}{constant composition distribution matching}
\newacronym{CRC}{CRC}{cyclic redundancy check}
\newacronym{diMC}{diMC}{discrete-input memoryless channel}
\newacronym{DMC}{DMC}{discrete memoryless channel}
\newacronym{DM}{DM}{distribution matching}
\newacronym{DPC}{DPC}{dirty paper coding}
\newacronym{dSNR}{dSNR}{design signal-to-noise ratio}
\newacronym{FEC}{FEC}{forward error control}
\newacronym{FER}{FER}{frame error rate}
\newacronym{GP}{GP}{Gelfand-Pinsker}
\newacronym{HY}{HY}{Honda-Yamamoto}
\newacronym{iff}{iff}{if and only if}
\newacronym{iid}{i.i.d.}{independent and identically distributed}
\newacronym{IM}{IM}{intensity modulation}
\newacronym{LDPC}{LDPC}{low-density parity-check}
\newacronym{LLPS}{LLPS}{linear layered probabilistic shaping}
\newacronym{LLR}{LLR}{log-likelihood ratio}
\newacronym{MAC}{MAC}{multi-access channel}
\newacronym{MC}{MC}{Monte Carlo}
\newacronym{MI}{MI}{mutual information}
\newacronym{MLC}{MLC}{multilevel coding}
\newacronym{MLHY}{MLHY}{multilevel Honda-Yamamoto}
\newacronym{MLPC}{MLPC}{multilevel polar coding}
\newacronym{MMSE}{MMSE}{minimum mean square error}
\newacronym{MSD}{MSD}{multistage decoding}
\newacronym{OOK}{OOK}{on-off keying}
\newacronym{PAM}{PAM}{pulse-amplitude modulation}
\newacronym{PAS}{PAS}{probabilistic amplitude shaping}
\newacronym{PS}{PS}{probabilistic shaping}
\newacronym{PCPAS}{PC-PAS}{polar-coded probabilistic amplitude shaping}
\newacronym{QAM}{QAM}{quadrature-amplitude modulation}
\newacronym{RCUB}{RCUB}{random coding union bound}
\newacronym{SCL}{SCL}{successive cancellation list}
\newacronym{SC}{SC}{successive cancellation}
\newacronym{SE}{SE}{spectral efficiency}
\newacronym{SIR}{SIR}{signal-to-interference ratio}
\newacronym{SMI}{SMI}{symmetric mutual information}
\newacronym{SNR}{SNR}{signal-to-noise ratio}
\newacronym{TCMPAS}{TCM-PAS}{trellis-coded modulation probabilistic amplitude shaping}
\newacronym{TSA}{TSA}{time-shifted alternating}
\newacronym{wlog}{w.l.o.g.}{without loss of generality}

\makeatletter
\input{tikz/output-figure8.dpth}
\makeatother

\begin{document}

\title{\myTitle}

\author{%
  \IEEEauthorblockN{Constantin Runge\,\orcidlink{0000-0001-8324-3945} and Gerhard Kramer\,\orcidlink{0000-0002-3904-9181}}
  \IEEEauthorblockA{Institute for Communications Engineering, %
  Technical University of Munich, 80333 Munich, Germany\\
  \{constantin.runge, gerhard.kramer\}@tum.de}}
  
\maketitle

\begin{abstract}
A coding scheme for \glspl{BC} is proposed that shifts the users' code blocks by different amounts of time and applies alternating Gelfand-Pinsker encoding. The scheme achieves all rate tuples in Marton's region for two receiver \glspl{BC} without time-sharing or rate-splitting. Simulations with short polar codes show that the method reduces the gap to capacity as compared to time-sharing.
\end{abstract}
\glsresetall

\section{Introduction}

\Glspl{BC} model the downlink of wireless cellular systems. A practical approach to avoid interference is orthogonalizing transmission, e.g., by time-division multiplexing, frequency-division multiplexing, or inverse precoding. The best-known rates for \glspl{BC}, up to multi-letter coding that is considered impractical, are achieved with Marton coding which simultaneously bins two or more random codebooks; see~\cite[p.~259]{Csiszar81}, \cite{Kramer03}, \cite{Marton79}. However, Marton coding also seems impractical, and implementations use binning with individual codebooks, called \gls{GP} coding \cite{Gelfand79}.
Capacity-achieving \gls{GP} coding can be implemented by polar codes~\cite{Arikan09} for joint shaping and error control~\cite{Honda13}.
For example, the authors of~\cite{Korada10b} apply polar codes to binary symmetric \gls{GP} channels; the paper~\cite{Arikan12} polarizes two random variables concurrently; the paper \cite{Goela15} uses polar codes for probabilistic shaping and \gls{GP} coding; the authors of~\cite{Mondelli15} use a chaining construction; the thesis~\cite{Liu16} uses polar lattices; and the papers~\cite{Sener21,Sener22,Sener24} use scalar lattices and probabilistic shaping.

\gls{GP} coding for \glspl{BC} achieves the corner points of Marton's region. The other rate points may be achieved by time-sharing or rate-splitting~\cite{Carleial78}, similar to \glspl{MAC} with \gls{SC} decoding.
For \glspl{MAC}, a simple scheme has the transmitters use time-shifted encoding~\cite{Hou06} and \gls{SC} decoding. This approach has a lower delay than time-sharing in general and is simpler than rate-splitting. A related idea is block-offset encoding and \gls{SC} decoding, which can improve rates, e.g., for multi-access relay channels~\cite{Sankar06} and noisy network coding~\cite{Ya11,Lim11}.

This paper studies a dual of the \gls{MAC} scheme in~\cite{Hou06}, which we call \gls{TSA} \gls{GP} encoding. Simulations suggest that polar codes are well-suited for the approach. The paper is organized as follows. Section~\ref{sec:prelim} reviews notation, Marton's region, and \gls{GP} coding. Section~\ref{sec:interlaced} performs \gls{TSA}-\gls{GP} encoding using random codes, and Section~\ref{sec:polar} uses polar codes. Section~\ref{sec:conclusions} concludes the paper.

\section{Preliminaries}
\label{sec:prelim}
\subsection{Notation}
\label{subsec:notation}

Random variables are written in upper case, such as $X$.
The alphabet, distribution, and realization of $X$ are written as $\mathcal X$, $P_X$ and $x$, respectively.
If $P_{Y_1|X}$ is stochastically degraded with respect to $P_{Y_2|X}$ we write $P_{Y_1|X} \preceq P_{Y_2|X}$.
An index set is denoted by $\idxset{n} \coloneqq \set{1,\dots,n}$.
Strings $(x_1,\dots,x_n)$ of symbols are denoted as $x^n$.
String concatenation is denoted as $[a^m, b^n]$ and $g^n(x^n) := (g(x_1),\dots,g(x_n))$.

The expressions $\E[X]$, $\He(X)$, $\He(X|Y)$, and $\MI(X; Y)$ refer to the expectation of $X$, the entropy of $X$, the conditional entropy of $X$ given $Y$, and the \gls{MI} of $X$ and $Y$, respectively. The binary entropy function is denoted $h_2(p)=-p\log_2 p - (1-p)\log_2(1-p)$ for $0<p<1$, and $h_2(0)=h_2(1)=0$.
The conditional Bhattacharyya parameter is defined as (see~\cite{Arikan10})
\begin{equation}
    Z(X|Y) = 2\E\mleft[\sqrt{P_{X|Y}(0|Y)P_{X|Y}(1|Y)}\,\mright]
\end{equation}
and satisfies (see~\cite{Arikan10}, \cite[Lemma~6]{Liu19})
\begin{align}
    Z(X|Y)^2 & \le \He(X|Y) \le Z(X|Y) \\
    Z(X|Y,S) & \le Z(X|Y) \label{eq:bhattacharyya_conditional}
\end{align}
where $X,Y,S \sim P_{XYS}$.
The string $x^n$ is said to be $\epsilon$-typical with respect to $P_X$ if
\begin{align*}
    \left| N(a|x^n)/n - P_X(a) \right| \le \epsilon P_X(a), \quad \text{for all $a\in\mathcal X$}
\end{align*}
where $N(a|x^n)$ is the number of times the letter $a$ occurs in the string $x^n$.
The set of length-$n$ $\epsilon$-typical strings with respect to $P_X$ is denoted as $\mathcal T_\epsilon^n(P_X)$. We write $P_X^n=(P_X)^n$ for \gls{iid} strings.

\edef\thefigurebackup{\thefigure}
\setcounter{figure}{1}
\begin{figure*}[!t]
    \centering
	\scalebox{1.2}{\includegraphics{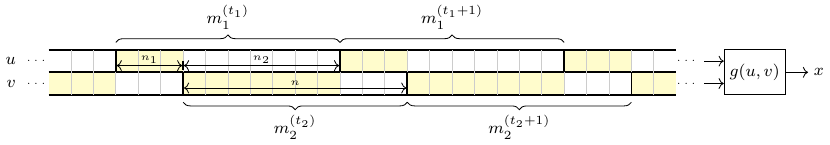}}
    \caption{\Gls{TSA}-\gls{GP} coding. The yellow region signifies \gls{GP} encoding, and the white region point-to-point encoding.}
    \label{fig:scheme}
    \hrulefill
\end{figure*}
\setcounter{figure}{\thefigurebackup}

\subsection{\texorpdfstring{\glspl{BC}}{}}

A two-receiver discrete memoryless \gls{BC} \cite{Cover72} has a conditional distribution $P_{Y_1Y_2|X}$ with input $X$ and two outputs $Y_1$, $Y_2$.
An $(n, R_1, R_2)$ code is a triple consisting of one encoder and two decoders where the encoder maps a message pair $(m_1,m_2) \in \idxset{2^{nR_1}}\times\idxset{2^{nR_2}}$ to a transmit string $x^n$ and each decoder $k\in\set{1,2}$ maps its channel observation $y_k^n \in \mathcal Y_k^n$ to an estimate $\hat m_k\in \idxset{2^{nR_k}}$.
Marton's region of rate pairs $(R_1,R_2)$ for a fixed $P_{UVX}$ is (see \cite{Marton79,ElGamal81} and also~\cite{Liang11})
\begin{gather}
    R_1 \le \MI(U; Y_1), \quad
    R_2 \le \MI(V; Y_2), \label{eq:Marton12} \\
    R_1 + R_2 \le \MI(U; Y_1) + \MI(V; Y_2) - \MI(U; V) \label{eq:Marton3}
\end{gather}
where $U$, $V$ are auxiliary random variables and the chain $(U,V)-X-(Y_1,Y_2)$ is Markov.
It suffices to choose $X=g(U,V)$ for some function $g : \mathcal U \times \mathcal V \to \mathcal X$.

\subsection{Achievability via \texorpdfstring{\Gls{GP}}{} Coding}

The interesting corner points of \eqref{eq:Marton12}-\eqref{eq:Marton3} are
\begin{align}
    (R_1,R_2) & = (\MI(U, Y_1), \MI(V; Y_2) - \MI(V; U)) \label{eq:corner1} \\
    (R_1,R_2) & = (\MI(U; Y_1) - \MI(U; V), \MI(V; Y_2) \textnormal{.}\label{eq:corner2}
\end{align}
Fig.~\ref{fig:blackwell_capacity} shows the capacity region of the Blackwell channel \cite{vanderMeulen75}, which can be achieved with \gls{GP} encoding; see~\cite{Marton79,Pinsker78}.
For the corner point \eqref{eq:corner1}, $m_1$ is encoded to $u^n$ with a usual encoder for the channel $P_{Y_1|U}$.
Next, $u^n$ is interpreted as side-information available at the encoder, and $m_2$ is encoded to $v^n$ using \gls{GP} coding for the channel $P_{Y_2|U,V}$ with state $U$.
Finally, the symbols $x_i = g(u_i,v_i)$, $i\in \idxset{n}$, are transmitted.
The corner point \eqref{eq:corner2} is achieved by swapping the roles of messages 1 and 2.

Time-sharing can achieve any point in Marton's region by transmitting at each corner point for some fraction of the time.
However, if the reliability constraints require similar code word lengths for both receivers, the delay is approximately a multiple of the code word lengths. This makes time-sharing impractical for low-delay transmissions.
On the other hand, rate-splitting splits one message into two messages and treats the \gls{BC} as a three-receiver channel, which incurs additional complexity for code design and at the encoder and decoders.

\begin{figure}
    \centering
	\scalebox{1.15}{\includegraphics{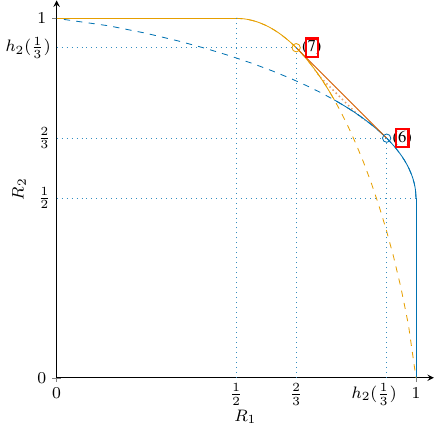}}
    \caption{Capacity region of the Blackwell channel \cite{vanderMeulen75}. The blue and yellow regions are achievable with conventional \gls{GP} coding, each corresponding to one encoding order.}
    \label{fig:blackwell_capacity}
\end{figure}

\section{\texorpdfstring{\Gls{TSA}}{} Encoding with Random Codes}
\label{sec:interlaced}

Time-sharing has each code word generated by a usual encoder or a \gls{GP} encoder.
\Gls{TSA} coding instead introduces an offset between the receiver blocks.
The encoders for each receiver perform \gls{GP} encoding in the first part of their blocks, where the interference of the other receiver's codeword is known, and a usual encoding in the second part; see Fig.~\ref{fig:scheme}.

\stepcounter{figure}

Again, two encoders produce $u^n$ and $v^n$.
Let $n_1 < n$ be the shift of the blocks of receiver 2.
Then, when $m_1$ is encoded, the last $n_1$ symbols of the previous $v^n$ overlap with the first $n_1$ symbols of the $u^n$ to be encoded.
The encoder for receiver 1 assumes the first $n_1$ states are known and treats the remaining $n - n_1$ states as unknown.
Similarly, the encoder for receiver 2 assumes the first $n - n_1$ states are known and treats the remaining $n_1$ states as unknown.

\begin{thm}
    \gls{TSA}-\gls{GP} encoding achieves all rate points in Marton's region.
\end{thm}
\begin{proof}
For the \gls{BC} $P_{Y_1Y_2|X}$ fix the distribution $P_{UV}$, a function $g: \mathcal U \times \mathcal V \to \mathcal X$, as well as overlap lengths $0 \leq n_1 \leq n$ and $n_2 = n - n_1$.
We transmit messages $m_1 \in \idxset{2^{nR_1}}$, $m_2 \in \idxset{2^{nR_2}}$ to receivers 1 and 2, respectively.
For $n_1 = 0$ or $n_2 = 0$, \gls{TSA} coding is identical to \gls{GP} coding for the corner points.
We assume $n_1,n_2 > 0$ in the following.

\paragraph*{Code Construction}
For receiver 1, choose $2^{nR_1 + n_1R'_1}$ \gls{iid} code words $u^{n_1}(m_1,l_1)$, $l_1\in\idxset{2^{n_1R'_1}}$, with $P_U^{n_1}$ and $2^{nR_1}$ \gls{iid} $u^{n_2}(m_1)$ with $P_U^{n_2}$.
For receiver 2, choose $2^{nR_2 + n_2R'_2}$ \gls{iid} $v^{n_2}(m_2,l_2)$, $l_2\in\idxset{2^{n_2R'_2}}$, with $P_V^{n_2}$ and $2^{nR_2}$ \gls{iid} $v^{n_1}(m_2)$ with $P_V^{n_1}$.

\paragraph*{Encoding}
The encoder maps $m_1$, $m_2$ to pairs $(u^{n_1}, u^{n_2})$, $(v^{n_2}, v^{n_1})$, respectively, in an alternating manner.
Given a previous $m_2$ and a current $m_1$, the encoder chooses an $l_1$ such that $(u^{n_1}(m_1, l_1), v^{n_1}(m_2)) \in \mathcal T_\epsilon^{n_1}(P_{UV})$ and transmits $x^{n_1} = g^{n_1}(u^{n_1}(m_1, l_1), v^{n_1}(m_2))$.
Simliarly, given a previous $m_1$ and a current $m_2$, the roles are reversed, i.e., the encoder chooses a $l_2$ such that $(u^{n_2}(m_1), v^{n_2}(m_2, l_2)) \in \mathcal T_\epsilon^{n_2}(P_{UV})$ and transmits $x^{n_2} = g^{n_2}(u^{n_2}(m_1), v^{n_2}(m_2, l_2))$.
In both cases, if there is no jointly typical string, set $l_k = 1$, $k=1,2$.

\paragraph*{Decoding}
Given $y_1^n = [y_1^{n_1}, y_1^{n_2}]$, receiver 1 finds indices $\hat m_1$, $\hat l_1$ such that $(u^{n_1}(\hat m_1, \hat l_1), y_1^{n_1}) \in \mathcal T_\epsilon^{n_1}(P_{UY})$ and $(u^{n_2}(\hat m_1), y_1^{n_2}) \in \mathcal T_\epsilon^{n_2}(P_{UY})$.
Receiver 2 proceeds analogously with $y_2^n = [y_2^{n_2}, y_2^{n_1}]$.

\paragraph*{Analysis}
The encoder uses standard \gls{GP} encoders, which are likely to succeed if $R_k' > \MI(U;V)$ and $n_k$ is large for $k=1,2$.
Suppose $(u^{n_1}, v^{n_1}, y_1^{n_1}) \in \mathcal T_\epsilon^{n_1}(P_{UVY_1})$ and consider the event $\mathcal E_1$ of finding wrong estimates $\hat m_1$, $\hat l_1$.
We have
\begin{align}
    \Pr(\mathcal E_1) &=    \sum_{\substack{\hat m, \hat l \\ \hat m \neq m_1}} \Pr(\mathcal E^1_\mathrm{GP}(\hat m, \hat l)) \Pr(\mathcal E^1_\mathrm{direct}(\hat m)) \\
    &\leq \sum_{\substack{\hat m, \hat l \\ \hat m \neq m_1}} 2^{-(n_1 \MI(U; Y_1) + n_2 \MI(U; Y_1)) + \delta_1(\epsilon)} \\
    &\leq 2^{n(R_1 + \frac{n_1}{n}R'_1 - \MI(U; Y_1)) + \delta_1(\epsilon)}
\end{align}
where the events $\mathcal E^1_\mathrm{GP}(m, l) = \set{(u^{n_1}(m, l), y_1^{n_1}) \in \mathcal T_\epsilon^{n_1}(P_{UY_1})}$ and $\mathcal E^1_\mathrm{direct}(m) = \set{(u^{n_2}(m), y_1^{n_2}) \in \mathcal T_\epsilon^{n_2}(P_{UY_1}}$ are independent, and $\lim_{\epsilon\to0}\delta_1(\epsilon) = 0$.
For the event $\mathcal E_2$ of finding erroneous indices at receiver 2, we similarly have
\begin{equation}
    \Pr(\mathcal E_2) \leq 2^{n(R_2 + \frac{n_2}{n}R'_2 - \MI(V; Y_2)) + \delta_2(\epsilon)} \textnormal{.}
\end{equation}
Defining $\alpha \coloneqq \frac{n_1}{n}$ and combining our results we obtain an achievable region of
\begin{align}
    R_1 &< \MI(U; Y_1) - \alpha\MI(U;V) \\
    R_2 &< \MI(V; Y_2) - (1-\alpha)\MI(U;V)
\end{align}
with $R_1 + R_2 < \MI(U; Y_1) + \MI(V; Y_2) - \MI(U; V)$.
\end{proof}

The above scheme requires no time-sharing or rate-splitting, and few blocks are needed to approach a desired rate tuple. For example, consider the target fraction $\alpha^* = a/N$ where $a$ and $N$ are integers.
With time-sharing, the average rates are close to the desired rate tuple only if the number of blocks is a multiple of $N$ or significantly larger than $N$.
Thus, the number of data bits for receiver $k$ with time-sharing is a multiple of $N \cdot n R_k$.
On the other hand, \gls{TSA} coding requires $N$ to be a divisor of the block length $n$, i.e., the number of data bits is a multiple of $n R_k$, without the extra factor $N$.
This difference is important for large $N$, e.g., when a fine-grained control of data rates is needed.
If $N$ does not divide $n$, then $\alpha^*$ can be closely approximated by some fraction $\frac{n_1}{n}$ for most practical block lengths.
Also, time-sharing and rate-splitting require multiple rates for different blocks, while the \gls{TSA} rates are the same for each block which simplifies coding.

\section{\texorpdfstring{\Gls{TSA}}{} Encoding With Polar Codes}
\label{sec:polar}

We demonstrate the practicality of \gls{TSA} coding via polar codes. Polar codes achieve the corner points of Marton's region using the schemes from \cite{Goela15,Mondelli15}. \gls{TSA} polar encoding modifies these schemes. %

\subsection{Polar Codes}

Polar codes are linear block codes defined via the self-inverse, linear polar transform $\matr G_n$ with
\begin{equation}
    \bar x^n = x^n \matr G_n^{-1} \textnormal{,\quad } \matr G_n = \bmat{1 & 0 \\ 1 & 1}^{\otimes \log_2 n} \label{eq:polar_transform} %
\end{equation}
where $\matr F^{\otimes k}$ is the $k$-fold Kronecker product of $\matr F$.
For strings $X^n$ and $Y^n$ we introduce the notation
\begin{align}
    \mathcal L_{X|Y}  &\triangleq \set{i\in\idxset n \mid Z(\bar X_i|\bar X^{i-1}, Y^n) < \delta_n} \label{eq:strongly_low} \\
    \mathcal H_{X|Y}  &\triangleq \set{i\in\idxset n \mid Z(\bar X_i|\bar X^{i-1}, Y^n) > 1-\delta_n}
\end{align}
for sets of strongly polarized positions and
\begin{align}
    \mathcal L_{X|Y}' &\triangleq \set{i\in\idxset n \mid Z(\bar X_i|\bar X^{i-1}, Y^n) \leq 1-\delta_n} \\
    \mathcal H_{X|Y}' &\triangleq \set{i\in\idxset n \mid Z(\bar X_i|\bar X^{i-1}, Y^n) \geq \delta_n} \label{eq:weakly_high}
\end{align}
for sets of strongly or weakly polarized positions
where $\bar X^n = X^n \matr G_n^{-1}$ is the polar transform of $X^n$ and $\delta_n = 2^{-n^\beta}$ for $0 < \beta < \frac12$.
These index sets polarize~\cite[Eqs.~(38), (39)]{Honda13}, i.e., we have
\begin{align}
    \lim_{n\to\infty} \frac1n\abs*{\mathcal L_{X|Y}} &= \lim_{n\to\infty} \frac1n\abs*{\mathcal L'_{X|Y}} = 1 - \He(X|Y) \label{eq:asym_polarization1} \\
    \lim_{n\to\infty} \frac1n\abs*{\mathcal H_{X|Y}} &= \lim_{n\to\infty} \frac1n\abs*{\mathcal H'_{X|Y}} = \He(X|Y)     \textnormal{.} \label{eq:asym_polarization2}
\end{align}
One usually has $X^n \sim \Ber^n(\frac12)$ for polar codes.

This polarization property can be used for coding by considering the polar transform $\bar x^n$ of the code word $x^n$ as follows. The bits $\bar x_i$ on positions $i \in \mathcal L_{X|Y}$, can be estimated reliably given previous $\bar x_j$, $j \in \idxset{i-1}$.
An encoder thus places data into the positions in $\mathcal L_{X|Y}$.
The remaining, so-called frozen, bits $\bar x_i$, $i \in \mathcal H'_{X|Y}$, are fixed to $0$.
An \gls{SC} decoder successively computes estimates of $\bar x_i$ given $y^n$: at every position $i$, either $\bar x_i = 0$ is known, or $\bar x_i$ can be estimated reliably.
The error probability of the entire decoding procedure goes to $0$ as $\mathcal O(2^{-n^\beta})$ for $0 < \beta < \frac12$.

\subsection{Polar Codes for Probabilistic Shaping and \texorpdfstring{\Gls{GP}}{} Coding}

The polar transform can be used to construct codes that emulate a distribution $P_{X^n}=\prod_i P_{X_i}$ \cite{Honda13}.
Consider, e.g., the \gls{iid} case $P_{X^nY^n} = P_X^nP_{Y|X}^n$.
We analyze $Z(\bar X_i | \bar X^{i-1}, Y^n)$ and $Z(\bar X_i | \bar X^{i-1})$.
The bits $\bar x_i$ with $i\in\mathcal H_X$ have entropy close to $1$ and can carry one bit of information each.
The shaping bits $\bar x_i$ with $i\in\mathcal L'_X$, have entropy mostly significantly less than $1$.
This means $X^n \sim P_{X^n}$ induces a non-uniform distribution on these bits given $\bar x^{i-1}$.
Encoding is performed by first fixing $\bar x_i = 0$, $i\in\mathcal H'_{X|Y}$, and placing data into $\bar x_i$, $i \in \mathcal M$, with $\mathcal M = \mathcal L_{X|Y} \cap \mathcal H_X$.
An \gls{SC} decoder then computes probabilities and samples $\bar x_i \sim P_{\bar X_i|\bar X^{i-1}}$, $i\in\mathcal L'_X$.
The receiver is similar to the one for polar codes with uniform $X$.
The bits with $i \in \mathcal H_X$ are estimated using an \gls{SC} decoder given $y^n$.
The bits at $\mathcal L'_X$ must be decoded with the same \gls{SC} decoder and randomness as at the encoder.
Since $\mathcal L_X \subseteq \mathcal L_{X|Y}$, the rate is \cite{Honda13}
\begin{equation}
    \lim_{n\to\infty}\frac1n\abs{\mathcal M} = \He(X) - \He(X|Y)
\end{equation}
and the decoding error probability scales as $\mathcal O(2^{-n^\beta})$.

The above carries over to \gls{GP} coding with polar codes when a state $S^n$ is known at the encoder.
Assuming $P_{S|X} \preceq P_{Y|X}$, one can show that $\mathcal L_{X|S} \subseteq \mathcal L_{X|Y}$ \cite{Korada10b}.
The index sets $\mathcal H_{X|S}$ and $\mathcal L_{X|Y}$ thus align, and the coding scheme for probabilistic shaping can be applied by simply considering $X|S$ and $X|Y$ instead of $X$ and $X|Y$ \cite{Korada10b,Honda13,Goela15}.
This degradedness assumption can be lifted using the chaining methods of \cite{Korada10b, Mondelli15}.
An error probability scaling of $\mathcal O(2^{-n^\beta})$ is achieved with the message set $\mathcal M = \mathcal H_{X|S} \cap \mathcal L_{X|Y}$ and the rate
\begin{equation}
    \lim_{n\to\infty}\frac1n\abs{\mathcal M} = \He(X|S) - \He(X|Y).
\end{equation}
For simplicity, we assume $P_{S|X} \preceq P_{Y|X}$.

\subsection{Proof of Polarization for \texorpdfstring{\Gls{TSA}}{} Coding}

Polar codes achieve the corner points of Marton's region if $P_{U|V} \preceq P_{Y_2|V}$ \cite[Thm.~3]{Goela15}.
The scheme splits the target distribution $P_{UV}$ into $P_U P_{V|U}$ and performs shaping with frozen bits $\mathcal H'_{U|Y_1}$ and shaping bits $\mathcal L'_U$ for receiver 1, and then polar \gls{GP} coding with frozen bits $\mathcal H'_{V|Y_2}$ and shaping bits $\mathcal L'_{V|U}$ for receiver 2.
As outlined above, this scheme achieves the \gls{GP} points with the message sets $\mathcal M_1 = \mathcal L_{U|Y_1} \cap \mathcal H_U$ and $\mathcal M_2 = \mathcal L_{V|Y_2} \cap \mathcal H_{V|U}$, i.e., we have
\begin{align}
    \lim_{n\to\infty}\frac1n\abs{\mathcal M_1} &= \MI(U; Y_1) \\
    \lim_{n\to\infty}\frac1n\abs{\mathcal M_2} &= \MI(V; Y_2) - \MI(U; V) \textnormal{.}
\end{align}
Using this polarization property, one can prove vanishing error probabilities of a coding scheme with shared randomness.

One hurdle to proving polarization for the \gls{TSA}-\gls{GP} scheme is that the strings of \glspl{biDMC} to be transformed are not stationary and depend on the block length $n$.
We cannot identify the transformed channels after the same number of transformations for different $n$ with each other.
Instead, we prove polarization for an ensemble of polar codes concatenated with a uniform random block interleaver of length $n$.
This interleaver allows to treat the non-stationary sequence as stationary and \gls{iid} when considering its polar transform.
We can model the interleaver as randomly assigning \gls{GP} encoding or point-to-point encoding to each position $i\in\idxset{n}$.

\begin{thm}
    Consider a \gls{biDMC} $P_{Y|UV}$ with channel input $U$ and random state $V$.
    For each block of length $n$, the partial state $v^{n_1}$ is known to the transmitter at positions $\mathcal A \subset \idxset n$, $\abs{\mathcal A} = n_1$.
    If the positions $\mathcal A$ are chosen uniformly at random for each code block,
    then there exist polar encoders and decoders with rate approaching $\MI(U;Y) - \alpha\MI(U;V)$ and error probability approaching $0$ as $\mathcal O(2^{-n^\beta})$, $0 < \beta < \frac12$.
\end{thm}
\begin{proof}
    Consider $V,S,U,Y \sim P_V P_S P_{U|VS} P_{Y|UV}$ where $S$ takes on the value 1 if  $i\in\mathcal A$ and 0 otherwise.
    Let $P_S = \Ber(\alpha)$, $P_{U|VS}(u|v, 0) = \E\mleft[P_{U|V}(u|V)\mright] = P_U(u)$ and $P_{U|VS}(u|v,1) = P_{U|V}(u|v)$.
    To show polarization, define a random variable $V'$ with alphabet $\mathcal V' = \mathcal V \cup \set{\bot}$ where $\bot\not\in\mathcal V$ as $v': \mathcal V \times \mathcal S \to \mathcal V', (v, 1) \mapsto v, (v, 0) \mapsto \bot$. %
    The models $P_V P_S P_{U|VS} P_{Y|UV}$ and $P_V P_S P_{V'|VS} P_{U|V'} P_{Y|UV}$ yield identical joint distributions for $V,S,U,Y$ when $P_{U|V'}(u|v) = P_{U|V}(u|v)$, $v\in\mathcal V$, and $P_{U|V'}(u|\bot) = P_U(u)$.
    We have
    \begin{align}
        \lim_{n\to\infty}\abs{\mathcal L_{U|Y}}  &= 1 - \He(U|Y) \\
        \lim_{n\to\infty}\abs{\mathcal H_{U|V'}} &= \alpha\He(U|V) + (1-\alpha)\He(U)
    \end{align}
    and thus 
    \begin{align}
        \frac1n\abs{\mathcal M} &\to (1-\alpha)\MI(U;Y) + \alpha\left(\MI(U;Y) - \MI(U;V)\right) \\
                                &= \MI(U;Y) - \alpha\MI(U;V)
    \end{align}
    with $\mathcal M = \mathcal L_{U|Y} \cap \mathcal H_{U|V'}$.
    This last step requires $P_{V'|U} \preceq P_{Y|U}$ to ensure aligned polarization indices and $\mathcal L_{U|V'} \subseteq \mathcal L_{U|Y}$.
    If this assumption does not hold, one can use a chaining construction.
    The randomized encoders and decoders from \cite{Korada10b,Honda13,Goela15} prove the existence of an encoder and a decoder with message set $\mathcal M$ for which the error probability $\Pr(\mathcal E_\mathrm{iid})$ scales as $\mathcal O(2^{-n^\beta})$, $0 < \beta < \frac12$, when communicating over the channel $P_{Y|UV}$ with \gls{iid} side information $V'$.
    
    It remains to show this scheme is also viable when $S$ is not \gls{iid} $\Ber(\frac{n_1}{n})$ but of fixed type $p^* = (n_1, n_2)$.
    Let $\Pr(\mathcal E_p)$ be the error probability conditioned on $s^n$ being of type $p$, $\mathcal P_n$ be the set of all types of length $n$ on the binary alphabet, and $p_{s^n}$ be the type of $s^n$.
    The bounds \cite[Lemmas~II.1, II.2]{Csiszar98} yield $\Pr(p_{S^n} = p^*) \geq \frac1n$ and thus 
    \begin{equation}
        \Pr(\mathcal E_{p^*}) \leq n\sum_{p \in \mathcal P_n} \Pr(p_{S^n} = p) \Pr(\mathcal E_p) = n\Pr(\mathcal E_\mathrm{iid}) \textnormal{.}
    \end{equation}
    Therefore, the error probability when using the polar coding scheme constructed for \gls{iid} $S$ over the channel with partial side information of length $n_1$ and a random interleaver goes to $0$ as $\mathcal O(2^{-n^\beta})$, proving there exists an interleaver at least as good as this average.
\end{proof}

The above proves the existence of \gls{TSA} encoders and decoders, and positions at which the partial side information is available, such that the error probability can approach zero and the rate approach $\MI(U;Y_1) - \alpha\MI(U;V)$.
We obtain the same result for receiver 2 with rate approaching $\MI(V;Y_2) - (1-\alpha)\MI(U;V)$.
That is, following \cite{Chou15}, the randomized encoders and decoders can be replaced by deterministic ones, which leaves the interleavers as random.
Empirical results show the polarization speed depends on the interleaver.
Using no interleaver works well with the natural decoding order on $\bar x^n$ as defined by \eqref{eq:polar_transform}-\eqref{eq:weakly_high}.
We next demonstrate that \gls{TSA} polar coding without an interleaver can perform better than time-sharing when using this natural decoding order.

\subsection{Numerical Results over the Blackwell Channel}

Consider the Blackwell channel~\cite{vanderMeulen75}, which is a noiseless \gls{BC} with a ternary input and two binary outputs.
The channel has the mapping $\mathcal X \to \mathcal Y_1 \times \mathcal Y_2, 0 \mapsto (0, 0), 1 \mapsto (0, 1), 2 \mapsto (1, 0)$.
Being noiseless, its capacity region is known~\cite{Marton79,Pinsker78}.
The sum rate-optimal input distribution is uniform on $\mathcal X$ which implies $Y_1 = U \sim \Ber(\frac13)$, $Y_2 = V$, $V|(U=0) \sim \Ber(\frac12)$, and $V|(U=1) \sim \Ber(0)$.
The corner rates are $\SI[parse-numbers=false]{2/3}{\bpcu}$ for the \gls{GP} coded user and $\approx \SI{0.918}{\bpcu}$ for the marginally coded user with sum-rate $\log_2 3 \approx \SI{1.585}{\bpcu}$.
To compare \gls{TSA} polar coding with time-sharing, we choose $\alpha = \frac12$.

Fig.~\ref{fig:fer_blackwell} compares the total error probability of \gls{TSA} and time-shared polar coding under \gls{SCL} decoding \cite{Tal15} with list size $L=32$ and block length $n=2048$.
The codes are constructed by computing the entropies $\He(\bar U_i|\bar U^{i-1}, V'^n)$ for each user and using all bit channels with entropy above a specified threshold $1-\delta$ for data.
The order of the polarization stages, and thus the decoding order of $\bar x^n$, is chosen to apply the largest butterflies, i.e., the $(i, i + \frac{n}{2})$ butterflies, to $x^n$ first.
Using no interleaver and this decoding order, \gls{TSA} encoding \pltref{plt:fer_time_shifted} operates closer to the sum capacity than time sharing for the same error probability.
We compare with time-sharing and $n=1024$ \pltref{plt:fer_time_shared_short}, i.e., individual blocks have length $n=1024$, and time sharing with $n=2048$ \pltref{plt:fer_time_shared}.

\begin{figure}
    \centering
	\includegraphics{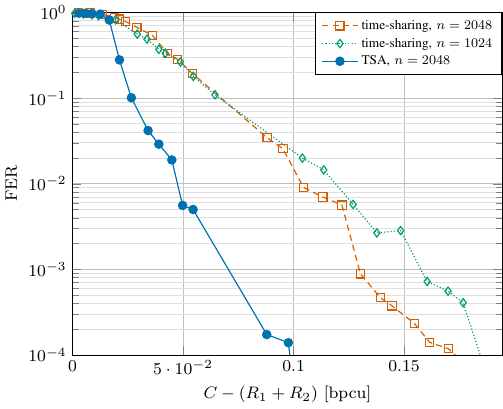}
    \caption{Simulated decoder error probabilities vs. back-off from sum capacity $C=\log_2 3$ bpcu over the Blackwell channel using SCL-32 decoding. The time-sharing fraction is $\alpha=1/2$ so $R_1=R_2$.}
    \label{fig:fer_blackwell}
\end{figure}

The improved performance is because the transformed channels of a $2\times2$ butterfly are at least as polarized as the initial channels.
That is, the first transformed channel is at least as unreliable as the less reliable of the initial two channels, and the second transformed channel is at least as reliable as the more reliable of the initial two channels \cite{Kim15}.
The degree of polarization of the transformed bit channels can be improved by pairing two channels of different reliability.
This is a useful heuristic to match parallel channels with different reliability to polar code symbols, c.f. \cite{Wiegart19a}.
By using a decoding order such that the indices $i$ and $i+\frac{n}{2}$ form a butterfly, time-shifting with $\alpha = \frac12$ and using the identity interleaver causes every butterfly in the first polarization stage to see two different channels.
For time-sharing, every butterfly in the first stage sees two identical channels.
This improves the degree to which the bit channels are polarized under time-shifting compared to time-sharing.

\section{Conclusions}
\label{sec:conclusions}

We presented \gls{TSA}-\gls{GP} encoding as a new coding technique for \glspl{BC}. The scheme can approach rate tuples between the Marton region corner points without time-sharing or rate-splitting.
The scheme increases the rates compared to time-sharing for short block lengths and has lower complexity than rate-splitting.
\gls{TSA} encoding can be implemented with polar codes, and numerical results confirm the anticipated rate gains. Further research may demonstrate the effectiveness of \gls{TSA} encoding over Gaussian \glspl{BC} and may consider tailored code designs to increase the gains.

\section*{Acknowledgement}

The authors thank Diego Lentner for valuable comments.
This work was supported by the Federal Ministry of Education and Research of Germany in the program of “Souverän. Digital. Vernetzt.”, joint project 6G-life, project identification number 16KISK002, and by
the German Research Foundation (DFG) under Project KR 3517/13-1.

\bibliographystyle{IEEEtran}
\bibliography{IEEEabrv,confs-jrnls,manual}

\end{document}